%% file: IEEE_TWCOM.tex
\documentclass[journal,draftclsnofoot,onecolumn,12pt,twoside]{IEEEtranTCOM}

\normalsize

\usepackage{cite}
\ifCLASSINFOpdf
\else
\fi

\usepackage[cmex10]{amsmath}

\usepackage[ruled,noline,linesnumbered]{algorithm2e}

\usepackage{import}
\usepackage{amsfonts}
\usepackage{amssymb}
\usepackage{theorem}
\usepackage[switch]{lineno} % to number text lines
\usepackage{graphicx}
\usepackage{xcolor}
\usepackage{subcaption}
\usepackage{float}

\def\C{{\bf C}}

\def\X{{\bf X}}
\def\x{{\bf x}}
\def\a{{\bf a}}
\def\z{{\bf z}}

\def\y{{\bf y}}
\def\n{{\bf n}}
\def\b{{\bf b}}
\def\w{{\bf w}}

\def\U{{\bf U}}

\def\V{{\bf V}}

\def\h{{\bf h}}
\def\I{{\bf I}}
\def\R{{\bf R}}

\def\W{{\bf W}}

\def\Y{{\bf Y}}

\def\X{{\bf X}}
\def\A{{\bf A}}
\def\B{{\bf B}}
\def\D{{\bf D}}

\def\Y{{\bf Y}}

\def\det{\operatorname{det}}

\def\diag{\operatorname{diag}}

\def\det{\operatorname{det}}

\def\deltaB{\boldsymbol{\delta}}
\def\He{\textnormal{H}}

\def\Gras{\mathbb{G}(M,{\mathbb{C}}^{T})}

{\theorembodyfont{\slshape}
\newtheorem{theorem}{Theorem}
\newtheorem{proposition}{Proposition}
\newtheorem{corollary}{Corollary}
\newtheorem{lemma}{Lemma}
}
{\theorembodyfont{\rmfamily}
\newtheorem{remark}{Remark}

}

\begin{document}

% \linenumbers

%\title{A Measure Preserving Mapping for Structured SIMO Grassmannian Constellations}
\title{Constellations on the Sphere with Efficient Encoding-Decoding for Noncoherent Communications}

\author{Diego~Cuevas, %~\IEEEmembership{Student~Member,~IEEE,}
        Javier~{\'A}lvarez-Vizoso,
        Carlos~Beltr{\'a}n,
        Ignacio~Santamaría, %~\IEEEmembership{Senior~Member,~IEEE}% <-this % stops a space
        V{\' i}t~Tu{\v c}ek
        and Gunnar~Peters%
\thanks{This work was supported by Huawei Technologies, Sweden under the project GRASSCOM. The work of D. Cuevas was also partly supported under grant FPU20/03563 funded by Ministerio de Universidades (MIU), Spain. The work of Carlos Beltr{\'a}n was also partly supported under grant PID2020-113887GB-I00 funded by MCIN/ AEI /10.13039/501100011033. The work of I. Santamaria was also partly supported under grant PID2019-104958RB-C43 (ADELE) funded by MCIN/ AEI /10.13039/501100011033. A short preliminary version of this paper was submitted to the 2022 IEEE Global Communications Conference: Signal Processing for Communications (Globecom 2022 SPC).}
\thanks{D. Cuevas, J. {\'A}lvarez Vizoso and I.Santamaria are with the Department
of Communications Engineering, Universidad de Cantabria, 39005 Santander, Spain (e-mail: diego.cuevas@unican.es; javier.alvarezvizoso@unican.es; i.santamaria@unican.es).}
\thanks{C. Beltr{\'a}n is with the Department of Mathematics, Statistics and Computing, Universidad de Cantabria, 39005 Santander, Spain (e-mail: carlos.beltran@unican.es).}
\thanks{V. Tucek and G. Peters are with the Department of Wireless Algorithms, Huawei Technologies, 16440 Kista, Sweden (email: vit.tucek@huawei.com; gunnar.peters@huawei.com)}}% <-this % stops a space


\markboth{Preprint}%
{Cuevas \textit{\lowercase{et al}}: Constellations on the Sphere with Efficient Encoding-Decoding for Noncoherent Communications}



\maketitle

\begin{abstract}
%\boldmath
In this paper, we propose a new structured Grassmannian constellation for noncoherent communications over single-input multiple-output (SIMO) Rayleigh block-fading channels. The constellation, which we call Grass-Lattice, is based on a measure preserving mapping from the unit hypercube to the Grassmannian of lines. The constellation structure allows for on-the-fly symbol generation, low-complexity decoding, and simple bit-to-symbol Gray coding. Simulation results show that Grass-Lattice has symbol and bit error rate performance close to that of a numerically optimized unstructured constellation, and is more power efficient than other structured constellations proposed in the literature and a coherent pilot-based scheme.
\end{abstract}


% \fi
\IEEEpeerreviewmaketitle



%\IEEEpubidadjcol

\input{section1_intro.tex}

\input{section2.tex}

\input{section3.tex}

\input{section4.tex}

\input{section5.tex}

\input{section6_conclusions.tex}

\appendix

\input{appendix1.tex}

%\subimport{./appendices}{appendix2.tex}

\input{appendix4.tex}

\input{appendix3.tex}

% \section*{Acknowledgment}


\ifCLASSOPTIONcaptionsoff
  \newpage
\fi



%\bibliographystyle{IEEEtranTCOM}

\bibliographystyle{ieeetr}
\bibliography{IEEE_TWCOM}

%\begin{biography}[{\includegraphics[width=1in,height=1.25in,clip,keepaspectratio]{mshell}}]{Michael Shell}


%\vfill


\end{document}

%% file: section1_intro.tex
 \section{Introduction}

\IEEEPARstart{I}{n} communications over fading channels, it is usually assumed that the channel state information (CSI) is typically estimated at the receiver side by periodic transmission of a few known pilots and then it is used for decoding at the receiver and/or for precoding at the transmitter. These are known as coherent schemes. The channel capacity for coherent systems is known to increase linearly with the minimum number of transmit and receive antennas at high signal-to-noise (SNR) ratio \cite{telatar,Foschini} when the channel remains approximately constant over a long coherence time (slowly fading scenarios). 

However, in fast fading scenarios or massive MIMO systems for ultra-reliable low-latency communications (URLLC), to obtain an accurate channel estimate would require pilots to occupy a disproportionate fraction of communication resources. These new scenarios that have emerged with 5G and B5G systems motivate the use of noncoherent communications schemes in which neither the transmitter nor the receiver have any knowledge about the instantaneous CSI.

Despite the receiver not having CSI, a significant fraction of the coherent capacity can be achieved in noncoherent communication systems when the signal-to-noise ratio (SNR) is high, as shown in \cite{marzetta_capacity, Hochwald00,tse_noncoherent,Conway96}. For the case of single-input multiple-output (SIMO) channels, which is the one we focus on in this paper, these works proved that at high SNR under additive Gaussian noise, assuming a Rayleigh block-fading SIMO channel with coherence time $T \geq 2$, the optimal strategy achieving the capacity is to transmit isotropically distributed unitary vectors belonging to the Grassmannian of lines or projective space \cite{tse_noncoherent,Conway96}. Equivalently, these constellations correspond to packings on the sphere. Therefore, in noncoherent SIMO communication systems the information is carried by the column span of the transmitted $T$-dimensional vector, ${\bf x}$, which is not affected by the SIMO channel ${\bf h}$. In other words, the column span of $\mathbf{x}$ is identical to the column span of $\mathbf{x} \mathbf{h}^{\textnormal{T}}$.
 
 An extensive research has been conducted on the design of noncoherent constellations as optimal packings on the Grassmann manifold \cite{Zhao04,dhillon2007constructing,Beko07,gohary,Cuevas21WSA,AlvarezEusipco22,Varanasi02,CuevasTCOM,Hochwald_TIT2000,Soleymani_TIT22,Hughes_TIT2000,Pitaval12,cipriano,cube_split}. Some experimental evaluation of Grassmannian constellations in noncoherent communications using over-the-air transmission has been reported in \cite{JacoboGrass}. 
 Existing constellation designs can be generically categorized into two groups: structured or unstructured. Among the unstructured designs we can mention the alternating projection method \cite{dhillon2007constructing}, the numerical methods in \cite{Beko07,gohary,Cuevas21WSA,AlvarezEusipco22}, which optimize certain distance measures on the Grassmannian (e.g., chordal or spectral), and the methods proposed in \cite{Varanasi02} and \cite{CuevasTCOM}, which maximize the so-called diversity product \cite{HanTIT06}.
 
On the other side, structured designs impose some kind of structure on the constellation points, facilitating low complexity constellation mapping and demapping. This is achieved through algebraic constructions such as the Fourier-based constellation in \cite{Hochwald_TIT2000} or the analog subspace codes recently proposed in \cite{Soleymani_TIT22}, designs based on group representations \cite{Hughes_TIT2000,Pitaval12}, parameterized mappings of unitary matrices such as the Exp-Map design in \cite{cipriano} or structured partitions of the Grassmannian like the recently proposed Cube-Split constellation \cite{cube_split}. The Cube-Split constellation is of particular interest for this work as it is the design most related to our proposal. Cube-Split is based on a mapping from the unit hypercube to the Grassmann manifold such that the constellation points are distributed approximately uniformly on the Grassmannian. However, the Cube-Split mapping only achieves uniformly distributed points for $T=2$. When $T>2$, Cube-Split ignores the statistical dependencies between the components of the codewords and applies the same mapping derived for $T=2$. These limitations are overcome with our proposed mapping, named Grass-Lattice, which is a {\it measure preserving mapping} between the unit hypercube and the Grassmannian for any value of $T\geq 2$. The fact that the Grass-Lattice mapping is measure preserving guarantees that any set of points uniformly distributed in the input space (the hypercube), is mapped onto another set of points or codewords uniformly distributed in the output space (the Grassmann manifold). The constellation structure allows for on-the-fly symbol generation, low-complexity decoding, and simple bit-to-symbol Gray coding.

This paper extends the work presented in \cite{Cuev2212:Measure}. The novelties are the following:

\begin{itemize}
    \item An alternative way of constructing vector $\w$ in mapping $\vartheta_2$ using a chi-squared random variable is presented.
    \item Mapping $\vartheta_3$ is now derived for any number of transmit antennas $M$, which is a first step to extend the Grass-Lattice mapping to the MIMO case.
    \item A visualization of the inputs and outputs of each mapping is provided for the case $T = 2$.
    \item More results showing the SER and BER performance as a function of parameter $\alpha$ are included.
    \item A new way of computing the optimum value of $\alpha$ based on the minimum chordal distance of the constellation is also proposed.
    \item We included as a baseline the performance of a coherent pilot-based scheme in terms of SER, BER, and spectral efficiency vs. $E_b / N_0$.
\end{itemize}

The remainder of this paper is organized as follows. The system model is presented in Section \ref{sec:preliminaries}. In Section \ref{sec:grasslattice} we describe the proposed measure preserving mapping, named Grass-Lattice, which maps points uniformly distributed in the unit hypercube to the Grassmann manifold $\mathbb{G}(1,\mathbb{C}^T)$. We next present the procedures for encoding and decoding using Grass-Lattice mapping in Section \ref{sec:encoding}. A comprehensive set of numerical simulation results to assess the performance of the proposed method in terms of symbol and bit error rates, as well as power efficiency, is provided in Section \ref{sec:results}. Finally, Section \ref{sec:conclusions} concludes the paper. In addition, the paper contains a set of appendices that include the proofs of the mathematical results. 
 
 \textit{Notation}: Matrices are denoted by bold-faced upper case letters, column vectors are denoted by bold-faced lower case letters, and scalars are denoted by light-faced lower case letters. The Euclidean norm is denoted by $\|v\|$ and $j$ denotes the imaginary unit. The superscripts $(\cdot)^{\textnormal{T}}$ and $(\cdot)^{\textnormal{H}}$ denote transpose and Hermitian conjugate, respectively. We denote by $\I_n$ the identity matrix of size $n$. A complex proper Gaussian distribution with zero mean and unit variance is denoted as ${\cal CN}(0,1)$ and $\x \sim {\cal CN}_{n}({\bf 0}, \R)$ denotes a complex Gaussian vector in $\mathbb{C}^n$ with zero mean and covariance matrix $\R$. For real variables we use $\x \sim {\cal N}_{n}({\bf 0}, \R)$. The complex Grassmann manifold of $M$-dimensional subspaces of the $T$-dimensional complex vector space $\mathbb{C}^T$ is denoted as $\Gras$. Particularly, the Grassmannian of lines $\mathbb{G}\left(1,\mathbb{C}^T\right)$, also called the complex projective space, is the space of one-dimensional subspaces in $\mathbb{C}^T$. Points in $\Gras$ are denoted as $[\X]$ and points in  $\mathbb{G}\left(1,\mathbb{C}^T\right)$ are denoted as $[\x]$.

%% file: section2.tex
\section{System Model}\label{sec:preliminaries}

\subsection{System Model}

We consider a noncoherent SIMO communication system where a single-antenna transmitter sends information to a receiver equipped with $N$ antennas over a frequency-flat block-fading channel with coherence time $T$ symbol periods. It is assumed that $T \geq 2$.
Hence, the channel vector $\mathbf{h} \in \mathbb{C}^{N}$ stays constant during each coherence block of $T$ symbols, and changes in the next block to an independent realization. The SIMO channel is assumed to be Rayleigh with no correlation at the receiver, i.e., $ \mathbf{h} \sim \mathcal{CN} \left( 0,\mathbf{I}_N \right)$, and unknown to both the transmitter and the receiver.

Within a coherence block the transmitter sends a signal $\mathbf{x} \in \mathbb{C}^{T}$, normalized as $\mathbf{x}^H\mathbf{x} = 1$, that is a unitary basis for the one-dimensional subspace $[\mathbf{x} ]$ in $\mathbb{G}\left(1,\mathbb{C}^T\right)$. The signal at the receiver $\mathbf{Y} \in \mathbb{C}^{T \times N}$ is 
\begin{equation}\label{eq:system_model}
    \mathbf{Y} = \mathbf{x} \mathbf{h}^{\textnormal T} + \sqrt{\frac{1}{T \rho}} \mathbf{W},
\end{equation}

\noindent where $\mathbf{W} \in \mathbb{C}^{T \times N}$ represents the additive Gaussian noise, with entries modeled as $w_{ij} \sim \mathcal{CN} \left( 0,1 \right)$, and $\rho$ represents the signal-to-noise-ratio (SNR).

In a noiseless situation, Grassmannian signaling guarantees error-free detection without CSI because $\mathbf{x}$ and the noise-free vector on a receive antenna $\mathbf{y} = \mathbf{x} h$ represent the same point in $\mathbb{G}(1,\mathbb{C}^T)$. 

For unstructured Grasmmannian codebooks, the optimal Maximum Likelihood (ML) detector (assuming equiprobable codewords) that minimizes the probability of error is given by

\begin{equation}\label{eq:ML}
    \Tilde{\mathbf{x}} = \arg \max_{\mathbf{x} \in \mathcal{C}} \| \mathbf{Y}^{\textnormal{H}} \mathbf{x} \|^2,
\end{equation}

\noindent where $\mathcal{C}$ represents the codebook of $K$ codewords. Each codeword carries $\log_2(K)$ bits of information.

The computational complexity of the ML detector increases with the number of codewords, $K$, since it is necessary to project the observation matrix onto each and every codeword. This is one of the main drawbacks of unstructured Grassmannian constellations especially when $K$ is high. Another drawback of unstructured codes is how to solve the bit labeling problem, for which there are generally only suboptimal or computationally intensive solutions. In the following section we present a structured Grassmannian constellation, called {\it Grass-Lattice}, which solves the two problems of unstructured constellations: it can be decoded efficiently with a computational cost that does not grow with $K$, and it allows for a Gray-like bit-to-symbol mapping function.

%% file: section3.tex
\section{Grass-Lattice Constellation}\label{sec:grasslattice}

\subsection{Overview}

The Grass-Lattice constellation for SIMO channels is based on a measure preserving mapping from the unit hypercube (product of the interval $(0,1)$ with itself ${2(T-1)}$ times) to the Grassmann manifold $\mathbb{G}\left(1,\mathbb{C}^T\right)$
\[
\vartheta:\mathcal{I}=\underbrace{(0,1)\times\cdots\times(0,1)}_{2(T-1) \, \text{times}}\to \mathbb{G}\left(1,\mathbb{C}^T\right),
\]
where recall that $T-1$ is the complex dimension of $\mathbb{G}\left(1,\mathbb{C}^T\right)$. Elements in $\mathcal{I}$ are denoted by
$$
(\a,\b) = (a_1,\ldots,a_{T-1},b_1,\ldots,b_{T-1}),\quad a_k,b_k\in(0,1).
$$

The Grass-Lattice mapping $\vartheta:\mathcal{I}\to \mathbb{G}\left(1,\mathbb{C}^T\right)$ has the following properties:
\begin{enumerate}
    \item The image of $\vartheta$ is all of $\mathbb{G}\left(1,\mathbb{C}^T\right)$ except for a zero--measure subset of $\mathbb{G}\left(1,\mathbb{C}^T\right)$,
    \item $\vartheta$ is a diffeomorphism onto its image,
    \item and the Jacobian of $\vartheta$ is constant.
\end{enumerate}

Given the mapping $\vartheta$, if we choose a set of input points uniformly distributed in the unit hypercube, the outputs points will be uniformly distributed in $\mathbb{G}\left(1,\mathbb{C}^T\right)$. The goal is to design structured codebooks that can be efficiently encoded (no need to store the constellation) and decoded (the real and imaginary parts $a_j,b_j$ can be decoded independently). To this end, we quantize the (0,1) interval with $2^B$ equispaced points, where $B \geq 1$ is the number of bits per real component, and generate a Grass-Lattice constellation with $|{\cal {C}}| = 2^{2(T-1) B}$ codewords. The rate of the code is $R = \frac{2(T-1) B}{T}$ b/s/Hz.

The Grass-Lattice mapping is composed of three consecutive mappings $\vartheta = \vartheta_3 \ \circ \ \vartheta_2 \ \circ \ \vartheta_1$, which are described in the following subsections.

\subsection{Mapping $\vartheta_1$}

Mapping $\vartheta_1$ maps points uniformly distributed in the unit hypercube $\mathcal{I}$ to points normally distributed in $\mathbb{C}^{T-1}$. The idea is to apply component-wise the inverse transform sampling method, which takes uniform samples on $[0,1]$ and returns the inverse of the cumulative distribution function with the desired distribution. More formally, we have the following classic result that is presented without proof.

\begin{lemma}\label{lem:twofold}
Let $a_k, b_k$ be independent random variables uniformly distributed in $[0,1]$: $a_k \sim {\cal{U}}[0,1]$ and $b_k \sim {\cal{U}}[0,1]$, and let $ z_k=F^{-1}(a_k) + j F^{-1}(b_k) $ where
\begin{equation}
\label{eq:cdf}
F(t)=\frac{1}{\sqrt\pi}\int_{-\infty}^te^{-s^2}\,ds.    
\end{equation}
Then, both $\Re(z_k)= F^{-1}(a_k)$ and $\Im(z_k)=F^{-1}(b_k) $ are independent Gaussian random variables that follow a $\mathcal N(0,1/2)$ distribution, and hence $z_k \sim \mathcal{CN}(0,1)$.
\end{lemma}

\subsection{Mapping $\vartheta_2$}

In  Lemma \ref{lem:pointsunitball} we describe the mapping $\vartheta_2$, which maps normally distributed points in $\mathbb{C}^{T-1}$ to points uniformly distributed in the unit ball
$$
\mathbb B_{\mathbb{C}^{T-1}}(0,1) = \{ \w \in \mathbb{C}^{T-1}, \|\w \|<1\}.
$$

\begin{lemma}\label{lem:pointsunitball}
Let $\z = \left(z_1,\ldots,z_{T-1}\right)^{\textnormal{T}}$ be a $(T-1)$-dimensional Gaussian vector with i.i.d. components $z_k\sim\mathcal{CN}(0,1)$. Moreover, let

\begin{equation}
    f_{T-1}(t) = \frac1t\left(\frac{2(T-1)}{\Gamma(T)}\int_0^ts^{2(T-1)-1}e^{-s^2}\,ds\right)^{1/(2(T-1))}
    =\frac1t\left(1-e^{-t^2}\sum_{k=0}^{T-2}\frac{t^{2k}}{k!}\right)^{1/(2(T-1))} \label{eq:fT}   
\end{equation}

\noindent Then, the random vector $ \w=\vartheta_2(\z)=\z f_{T-1}(\|\z\|)$
is uniformly distributed in the unit ball $\mathbb B_{\mathbb{C}^{T-1}}(0,1)$.  
\end{lemma}

\begin{proof}
The proof is given in Appendix \ref{app2}.
\end{proof}

\begin{remark}
Since $\z \sim {\cal{CN}} ({\bf 0}, \I_{T-1})$, then $2 \| \z \|^2 \sim \chi_{2(T-1)}^2$. The random vector $\w = \vartheta_2(\z) = \z f_{T-1}(\|\z\|)$ can be alternatively constructed as follows. Begin with the unit-norm vector $\z/ \| \z \|$ that lies on ${\cal{S}}^{2(T-1)-1}$, where $d = T - 1$, and scale it as
\begin{equation*}
    \w = \frac{\z}{\| \z \|} (F_{ \chi_{2(T-1)}^2}(2 \| \z \|^2))^{1/(2(T-1))},
\end{equation*}
where $F_{ \chi_{2(T-1)}^2}(y)$ is the cdf of a chi-squared random variable with $2(T-1)$ degrees of freedom evaluated at $y$, which can be computed in closed-form as
\[
F_{ \chi_{2(T-1)}^2}(y) = 1-e^{-y/2}\sum_{k=0}^{(T-1)-1}\frac{y^{k}}{2^k k!}.
\]
The distribution of the squared norm of $\w$ can be derived as follows
\begin{align*}
    \|\w\|^2 &= \frac{\z^H \z}{\| \z \|^2} (F_{ \chi_{2(T-1)}^2}(2 \| \z \|^2))^{1/(T-1)} \\
    &= (F_{ \chi_{2(T-1)}^2}(2 \| \z \|^2))^{1/(T-1)}.
\end{align*}
Since $2 \| \z \|^2 \sim \chi_{2(T-1)}^2$, then $F_{ \chi_{2(T-1)}^2}(2 \| \z \|^2)$ is uniformly distributed in $[0,1]$. It is a known property that if $x \sim U[0,1]$ then $x^{1/r} \sim {\rm Beta}(r,1)$. All together, this shows that $\|\w\|^2 \sim {\rm Beta}(T-1,1)$.
Finally, it is also interesting to point out that the integral expression in Lemma \ref{lem:pointsunitball} is a (lower) incomplete gamma function:
$$
\int_0^{t^2}s^{T-2}e^{-s}\,ds=\gamma(T-1,t^2).
$$
\end{remark}

\subsection{Mapping $\vartheta_3$}

In this section we present the mapping $\vartheta_3$, which maps uniformly distributed points in the unit ball $\mathbb B_{\mathbb{C}^{(T-M)\times M},op}(0,1)$ to points uniformly distributed in $\mathbb{G}\left(M,\mathbb{C}^T\right)$. We will first derive the mapping $\vartheta_3$ for any value of $M$ and then we will particularize it for $M = 1$.

\begin{lemma}\label{lem:otrojacobian}
Consider the mapping 
$$
\begin{matrix}
\Theta:&\mathbb C^{(T-M)\times M}&\to&\{ \W \in \mathbb{C}^{(T-M)\times M} \, , \, \| \W\|_{op} < 1 \} \\
&\A&\mapsto&\A(\I_M+\A^\He \A)^{-1/2},
\end{matrix}
$$
whose inverse is
$$
\begin{matrix}
\Theta^{-1}:&\{\W\in\mathbb C^{(T-M)\times M},\|\W\|_{op}<1\}&\to&\mathbb C^{(T-M)\times M}\\
&\W&\mapsto&\W(\I_M-\W^\He \W)^{-1/2}.
\end{matrix}
$$
Then, the Jacobian of $\Theta$ equals $\det(\I_M+\W^\He \W)^{-T}$.
\end{lemma}
\begin{proof}
The proof is given in Appendix \ref{app4}.
\end{proof}

We are ready to prove the following result:
\begin{proposition}\label{prop:balltoG}
For all integrable $f:\mathbb C^{(T-M)\times M}\to\mathbb{C}$ we have
\begin{equation}\label{eq:desired}
    \int_{\A\in \mathbb C^{(T-M)\times M}}\frac{f(\A)}{\det(\I_M+\A^\He\A)^T}\,d\A=
    \int_{\underset{\|\W\|_{op}<1}{\W\in \mathbb C^{(T-M)\times M}} }f\left(\W(\I_M-\W^\He\W)^{-1/2}\right)\,d\W,
\end{equation}
\end{proposition}
\begin{proof}
The proof follows from the change of variables theorem and Lemma \ref{lem:otrojacobian} above.
\end{proof}
We immediately get:
\begin{corollary}\label{cor:integrales2}
For all integrable $f:\Gras\to\mathbb{C}$ we have
$$
\int_{[\X]\in\Gras}f([\X])\,d[\X]= \int_{\underset{\|\W \|_{op}<1}{\W \in \mathbb C^{(T-M)\times M}} }f \left( \begin{bmatrix} \sqrt{\I_M-\W^\He\W } \\ \W \end{bmatrix} \right) \,d\W .
$$
\end{corollary}
In other words: {\em in order to generate a uniform random element $[\X]$ in $\Gras$, one may generate a random uniform element $\W $ in the operator norm unit ball of $\mathbb C^{(T-M)\times M}$ and output}
\[
\begin{bmatrix} \sqrt{\I_M-\W^\He\W } \\ \W \end{bmatrix}.
\]
%\begin{remark}
%Corollary \ref{cor:integrales2} can also be deduced from some theoretical results in symplectic geometry, since the mapping $\vartheta_3^{(M)}$ can be proved to be a symplectomorphism, see \cite{Karshon} for $M=1$ of  \cite[Lemma 4.1]{Lu} for general $M$. Symplectomorphisms are measure--preserving maps which yields an alternative proof to Corollary \ref{cor:integrales2}. However, our proof is direct and produces intermediate results with some interest such as Proposition \ref{prop:balltoG}.

%\end{remark}

%\begin{corollary}\label{cor:vol_of_op_norm_1}
%The volume of the set of complex $(T-M)\times M$ matrices of operator norm at most $1$ equals the volume of the Grassmannian $\Gras$.
%\end{corollary}
%\begin{proof}
%Use $f\equiv1$ in Corollary \ref{cor:integrales2}.
%\end{proof}

%%%%%%%%%%%%%%%%%%%%%%%%%%%%%
%%%%%%%%%%%%%%%%%%%%%%%%%%%%%
%%%%%%%%%%%%%%%%%%%%%%%%%%%%%

In Lemma \ref{cor:integralesM1} we particularize the mapping $\vartheta_3$ for $M = 1$, which maps uniformly distributed points in the unit ball $\mathbb B_{\mathbb{C}^{T-1}}(0,1)$ to points uniformly distributed in $\mathbb{G}\left(1,\mathbb{C}^T\right)$.

\begin{lemma}\label{cor:integralesM1}
The mapping 
$$
\begin{matrix}
\vartheta_3:&  \w \in \mathbb B_{\mathbb{C}^{T-1}}(0,1) &\to &\mathbb{G}\left(1,\mathbb{C}^T\right) \\
&\w&\mapsto&   \begin{bmatrix} \sqrt{1 -\|\w \|^2 } \\ \w \end{bmatrix}
\end{matrix}
$$
is measure preserving. So in order to generate a uniform random element $[\x]$ in $\mathbb{G}\left(1,\mathbb{C}^T\right)$, one may generate a random uniform element $\w$ in $\mathbb B_{\mathbb{C}^{T-1}}(0,1)$ and output $\left[\sqrt{1 -\|\w \|^2 }, \w^{\textnormal{T}}\right]^\textnormal{T}$. 
\end{lemma}

%\begin{proof}
%Use $M = 1$ in Corollary \ref{cor:integrales2}. 
%\end{proof}

\subsection{Main result}
The following theorem summarizes the measure preserving Grass-Lattice mapping for SIMO channels.

\begin{theorem}\label{th:genproj}
  Let us consider a noncoherent SIMO communication system with coherence time $T \geq2$ and let $(\a, \b) = (a_1,\ldots, a_{T-1}, b_1,\ldots, b_{T-1} )$ be any point in the unit hypercube ${\cal{I}}$. The mapping $\vartheta:\mathcal{I}\to \mathbb{G}\left(1,\mathbb{C}^T\right)$ given by
  \begin{equation*}
  \vartheta(\a , \b) = \begin{bmatrix} \sqrt{1-\|\w\|^2} \\ \w \end{bmatrix},
  \end{equation*}
\noindent where:
\begin{itemize}
    \item $\w =\z f_{T-1}(\|\z\|)$, where $f_{T-1}$ is defined in \eqref{eq:fT}.
     \item $\z=(z_1,\ldots,z_{T-1})^{\textnormal{T}}$ with $z_k=F^{-1}(a_k) +j F^{-1}(b_k)$, where $F(x)$ is given in \eqref{eq:cdf}.
\end{itemize}
 Then, $\vartheta$ has a constant Jacobian and thus it is measure preserving. 
\end{theorem}

\begin{proof}
The proof is given in Appendix \ref{app_th1}.
\end{proof}

%% file: section4.tex
\section{Encoding and Decoding}\label{sec:encoding}

%As the measure preserving map is defined on an open interval $(0,1)^{2\left(T-1\right)}$, for a given number $B$ of bits per real component, we consider $2^B$ equispaced points on the interval $[\alpha, 1-\alpha]$:

For $M=1$, the Grassmann manifold $\Gras$ has complex dimension $T-1$ and real dimension $2(T-1)$. Since the measure preserving map we define has domain $(0,1)^{2\left(T-1\right)}$ and $(0,1)$ is an open interval, whatever discretization we choose in $(0,1)$ will necessarily have a lowest point $\alpha>0$ and a highest point $1-\beta<1$. Due to the symmetry of the mapping we find no reasons to choose $\beta\neq\alpha$ and hence for a given number $B$ of bits per real component, we consider $2^B$ equispaced points on the interval $[\alpha, 1-\alpha]$:
\begin{equation}
    \hat{x}_p = \alpha + p \ \frac{1-2\alpha}{2^B-1}, \quad 0\leq p \leq 2^B-1,
    \label{eq:alpha}
\end{equation}
where $\alpha$ is a parameter that can be optimized for performance (see Sec. \ref{sec:results}). The discretization of the real and imaginary (I/Q) components as in \eqref{eq:alpha} allows us to use a simple bit-to-symbol Gray mapper. Therefore, the uniformly distributed points on the unit cube $a_1,b_1,\ldots,a_{T-1},b_{T-1}$ are chosen randomly from the regular lattice defined by  \eqref{eq:alpha}. The procedure for computing the codeword to be transmitted $\x$ for an input $a_1,b_1,\ldots,a_{T-1},b_{T-1}$ is then:
  \begin{enumerate}
    \item Compute $z_k = F^{-1}(a_k) + j F^{-1}(b_k)$, $k=1,\ldots, T-1$, where $F(x)$ is the cdf of a ${\cal{N}}(0,1/2)$. The point $\z$ is isotropically distributed as $\z \sim {\cal{CN}}({\bf 0}, \I_{T-1})$.
    \item Compute $\w =\z f_{T-1}(\|\z\|)$, where $f_{T-1}(\cdot)$ is given in \eqref{eq:fT}. The point $\w$ is uniformly distributed in $\mathbb{B}_{\mathbb{C}^{T-1}}(0,1)$.
  \item Output 
  $$\x= \begin{bmatrix} \sqrt{1-\|\w\|^2} \\ \w \end{bmatrix}.$$
  \end{enumerate}
  
  The result of this procedure is a point $[\x]$ with representative $\x = \left[ \sqrt{1 - \| \w \|^2}, \w^{\textnormal{T}} \right]^{\textnormal{T}}$ which is uniformly distributed in $\mathbb{G}\left(1,\mathbb{C}^T\right)$.
  The cardinality of the structured Grassmannian constellation is $|{\cal {C}}|= 2^{2B(T-1)}$, and the spectral efficiency or rate is $R = \frac{2 B (T-1)}{T} = 2B \left(1-\frac{1}{T}\right)$ b/s/Hz.
  
  The input and output of the mappings $\vartheta_1$, $\vartheta_2$ and $\vartheta_3$ that form the Grass-Lattice mapping can be plotted for the case $T = 2$. For this specific case, the input $(\a,\b)$ has two real components $(a_1,b_1)$ and vectors $\z$ and $\w$ have one complex component ($z_1$ and $w_1$ respectively). To represent the points $[\x]$ with representative $\x = \left[x_1, x_2\right]^{\textnormal{T}} =\left[\sqrt{1 - |w_1|^2},w_1\right]^\textnormal{T}$ we use the Hopf map:
  
  \begin{equation}
    \begin{matrix}
        p:&\{(x_1, x_2)\in \mathbb{R} \times \mathbb{C}:|x_1|^2+|x_2|^2=1\} &\to& \mathbb{S}^2 \\
        &(x_1, x_2)&\mapsto& \left(2x_1 x_2^*, |x_2|^2 - |x_1|^2\right).
    \end{matrix}
  \end{equation}
  
   Fig. \ref{fig:grasslattice_mapping_005} shows the generation of the whole Grass-Lattice constellation for $T = 2$, $B = 4$ and $\alpha = 0.05$.
  
\begin{figure}[t!]
    \begin{subfigure}[b]{.49\textwidth}
        \centering
        \includegraphics[width=\columnwidth]{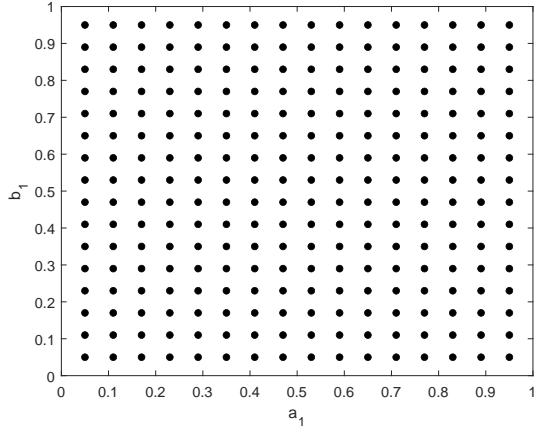}
        \caption{Lattice $\mathcal{I}$.}
        \end{subfigure}
        \begin{subfigure}[b]{.49\textwidth}
        \centering
        \includegraphics[width=\columnwidth]{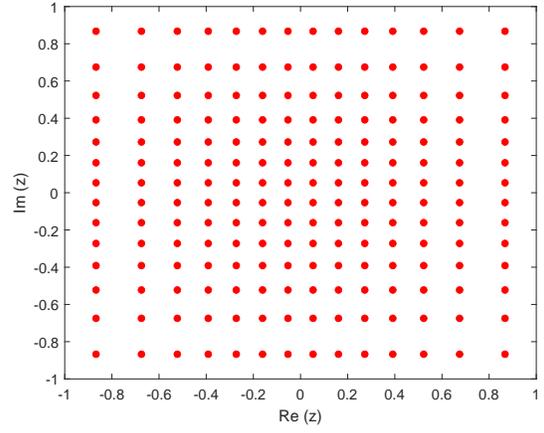}
        \caption{Normally distributed points in $\mathbb{C}$.}
    \end{subfigure}\hfill
    \begin{subfigure}[b]{.49\textwidth}
        \centering
        \includegraphics[width=\columnwidth]{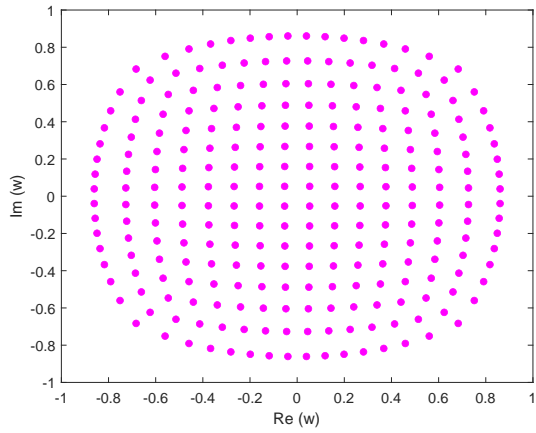}
        \caption{Uniformly distributed points in $\mathbb{B}_{\mathbb{C}}(0,1)$.}
    \end{subfigure}
    \begin{subfigure}[b]{.49\textwidth}
        \centering
        \includegraphics[width=\columnwidth]{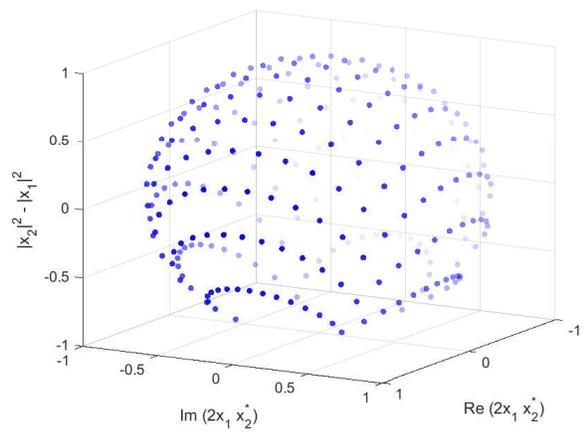}
        \caption{Uniformly distributed points in $\mathbb{G}(1,\mathbb{C}^2)$.}
    \end{subfigure}\hfill
    \caption{Grass-Lattice mapping for $T = 2$, $B = 4$ and $\alpha = 0.1$.}
    \label{fig:grasslattice_mapping_005}
\end{figure}

For the Grass-Lattice decoding, let us first consider the case where the number of receive antennas is $N=1$, so the received $T\times 1$ signal is $\y = \x h + \n$. Let $ \y = (v_0,\mathbf{v})$, then the decoder performs the following sequence of steps:
\begin{enumerate}
  \item Compute $\w=\mathbf{v} |v_0|/(v_0\|\y\|)$ (the chordal distance from $\left[\sqrt{1-\|\w\|^2} , \w^{\textnormal{T}}\right]^{\textnormal{T}}$ to $\y$ in $\mathbb G(1,\mathbb C^T)$ is minimal for this choice of $\w$).
  \item Solve the equation $sf_{T-1}(s)=\|\w\|$, for instance by bisection, and let $\z = s \w/\|\w\|$. Denote by $z_1,\ldots,z_{T-1}$ its complex components.

  \item Compute $\hat{a}_k= F(\Re(z_k))$, $\hat{b}_k= F(\Im(z_k))$, where $F(x)$ is the cdf of a ${\cal{N}}(0,1/2)$.
  \item Finally, $a_k = \lfloor \hat{a}_k \rceil $ and $b_k = \lfloor \hat{b}_k \rceil $ where $\lfloor x \rceil $ denotes the nearest point to $x$ in the lattice \eqref{eq:alpha}. 
\end{enumerate}

\paragraph*{\textbf{Multi-antenna receiver}} For $N > 1$, we just perform a denoising step at the decoder before doing steps 1-4 above. To do so, we use the fact that the signal of interest $\x \h^{\textnormal{T}}$ in \eqref{eq:system_model} is a rank-1 component of $\mathbf{Y}$. From the Eckart-Young theorem, the best rank-1 approximation in the Frobenius norm of $\Y$ is given by $ \lambda_1 \mathbf{r} \mathbf{g}^{\textnormal{H}}$, where $\lambda_1$ is the largest singular value of $\mathbf{Y}$, and $\mathbf{r}$ and $\mathbf{g}$ are the corresponding left and right singular vectors. We then take $\mathbf{r} = (v_0,\mathbf{v})$ as a denoised $T\times 1$ vector of observations and compute the sequence of steps 1-5 above. Interestingly, $\mathbf{r}$ is the solution of
\begin{equation*}
    \arg \max_{\mathbf{r} \in \mathbb{C}^T : \ \| r \|^2 = 1} \| \mathbf{Y}^{\textnormal{H}} \mathbf{r} \|^2,
\end{equation*}
so it can be viewed as a relaxed version of the ML decoder presented in \eqref{eq:ML} where the discrete nature of the constellation has been relaxed. Therefore, $\mathbf{r}$ is a rough estimate of the transmitted symbol $\mathbf{x}$ on the unit sphere.

The encoding and decoding for the Grass-Lattice constellation can be performed on the fly, without the need to store the entire constellation. At the decoder, after performing steps 1-4 above, the complexity is that of a symbol-by-symbol detector per real component, similar to the decoding of a QAM constellation.

%% file: section5.tex
\section{Performance Evaluation}\label{sec:results}

% \textcolor{red}{Include estimation of $\alpha_{opt}$ based on minimum chordal distance between codewords for small packings.}

In this section, we assess the performance of the proposed Grass-Lattice constellation, and compare it to other structured and unstructured Grassmannian constellations used for noncoherent communications. Since we compare constellations with different spectral efficiencies, we will show figures of SER or BER versus $E_b / N_0$ (SNR normalized by the spectral efficiency).

\subsection{SER/BER vs. $\alpha$}
%Let us first evaluate the influence of $\alpha$, which determines the length of the lattice used for each real component in \eqref{eq:alpha}, on the SER (same influence as on the BER). Fig. \ref{fig:alpha} shows the SER variation for a fixed SNR = $20$ dB, $T \in \{4,6\}$, $N = 2$ and $B \in \{1,2\}$. As we can see, $\alpha$ may have a significant impact on the SER performance of the Grass-Lattice constellation. Further, the SER varies significantly with the number of bits, $B$, used to encode each real component. For the scenarios considered in Fig. \ref{fig:alpha}, the optimal value of $\alpha$ is around 0.2 for $B = 1$ and around $0.15$ for $B = 2$. For the rest of experiments in this section, we will choose the value of $\alpha$ that provides the lowest SER at SNR = $20$ dB (the optimal value of $\alpha$ does not differ significantly for other SNRs). 

\begin{figure}[t!]
    \begin{subfigure}{0.49\textwidth}
        \centering
        \includegraphics[width=\columnwidth]{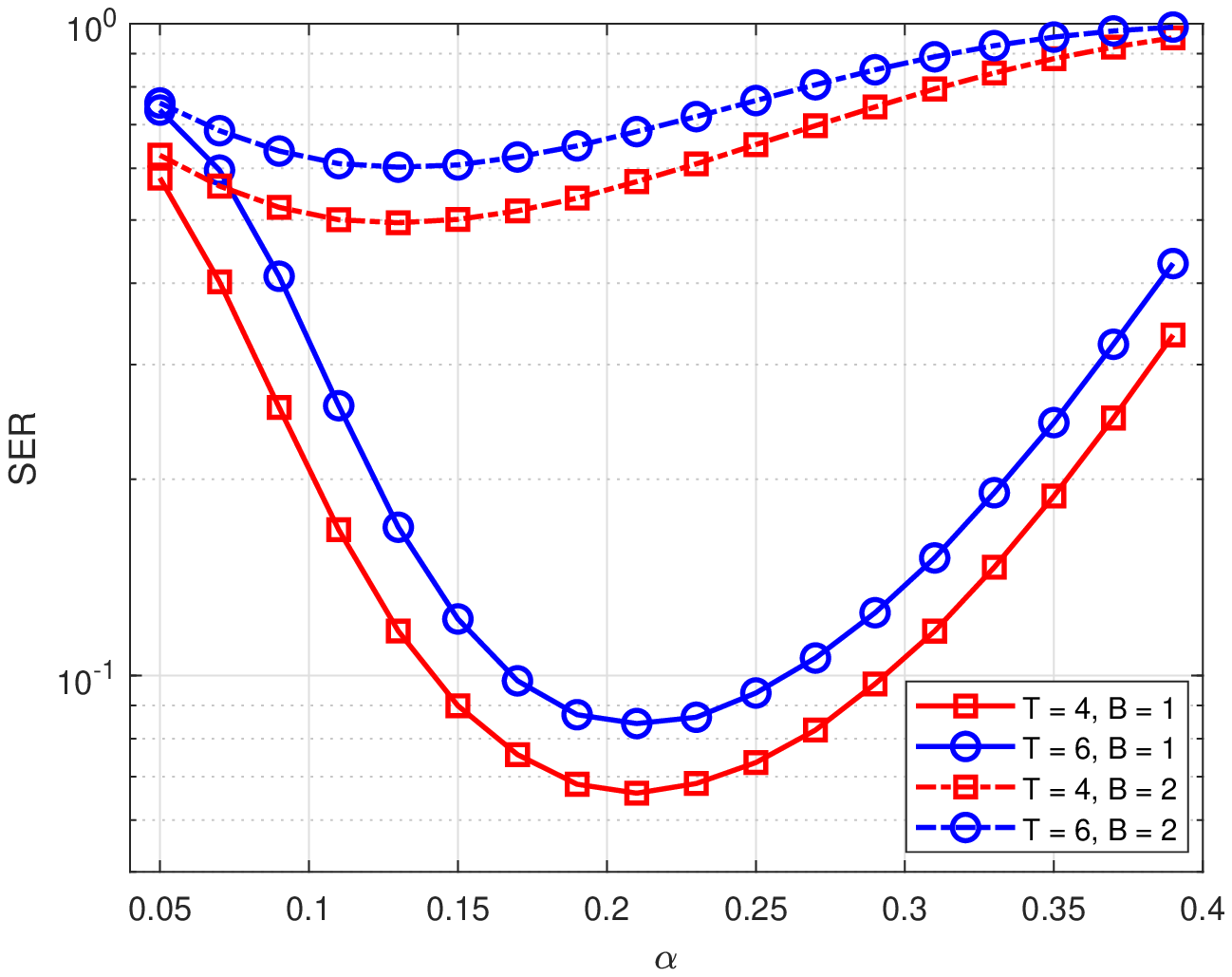}
        \caption{SNR = 10 dB}
        \label{fig:alpha_SNR10}
        \end{subfigure}
        \begin{subfigure}{0.49\textwidth}
        \centering
        \includegraphics[width=\columnwidth]{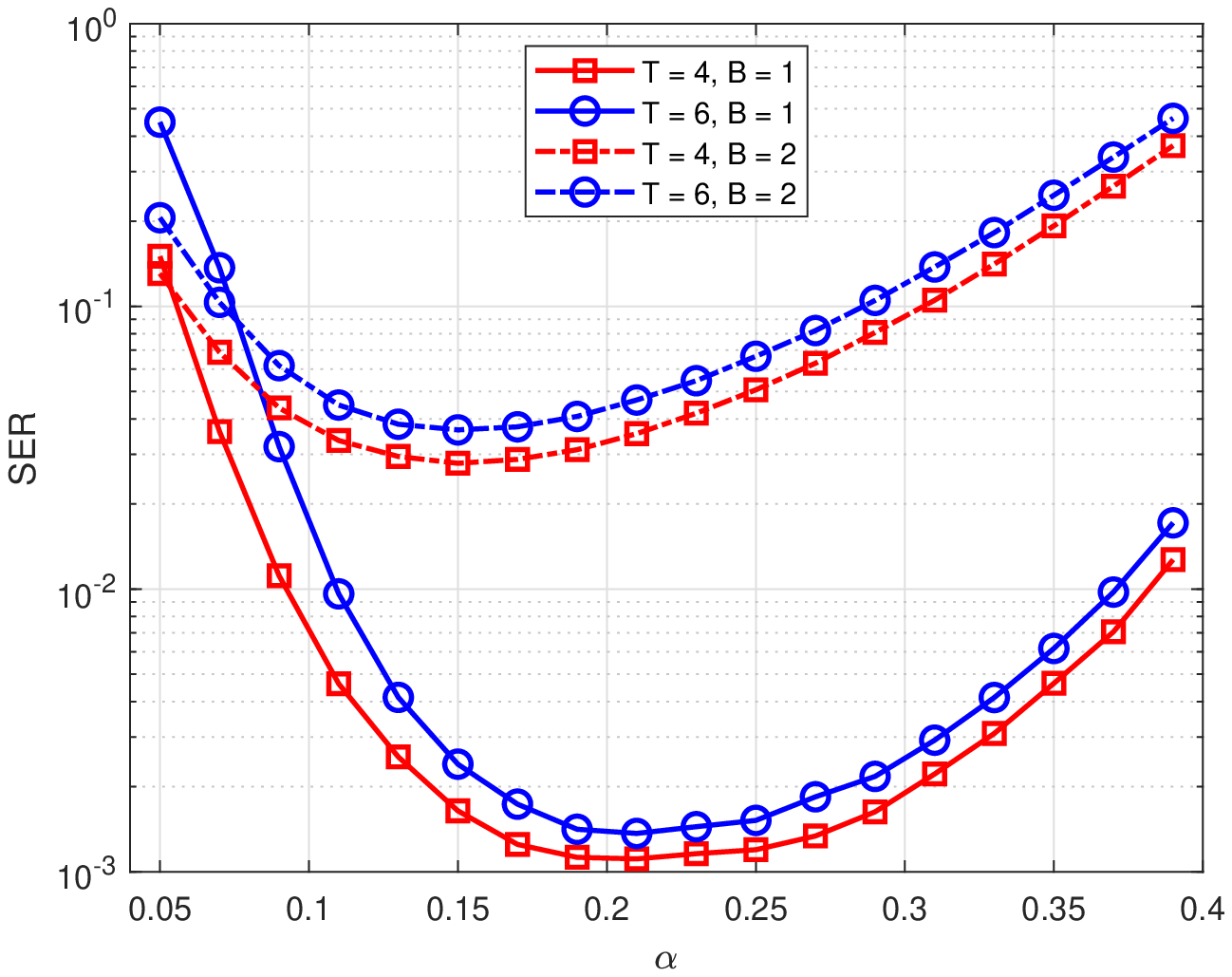}
        \caption{SNR = 20 dB}
        \label{fig:alpha}
    \end{subfigure}\hfill
    \caption{SER as a function of $\alpha$ of the Grass-Lattice constellation for $T \in \{4,6\}$, $N = 2$, $B \in \{1,2\}$ and SNR $\in \{10,20\}$ dB.}
\end{figure}

\begin{figure}[h!]
    \begin{subfigure}{0.49\textwidth}
        \centering
        \includegraphics[width=\columnwidth]{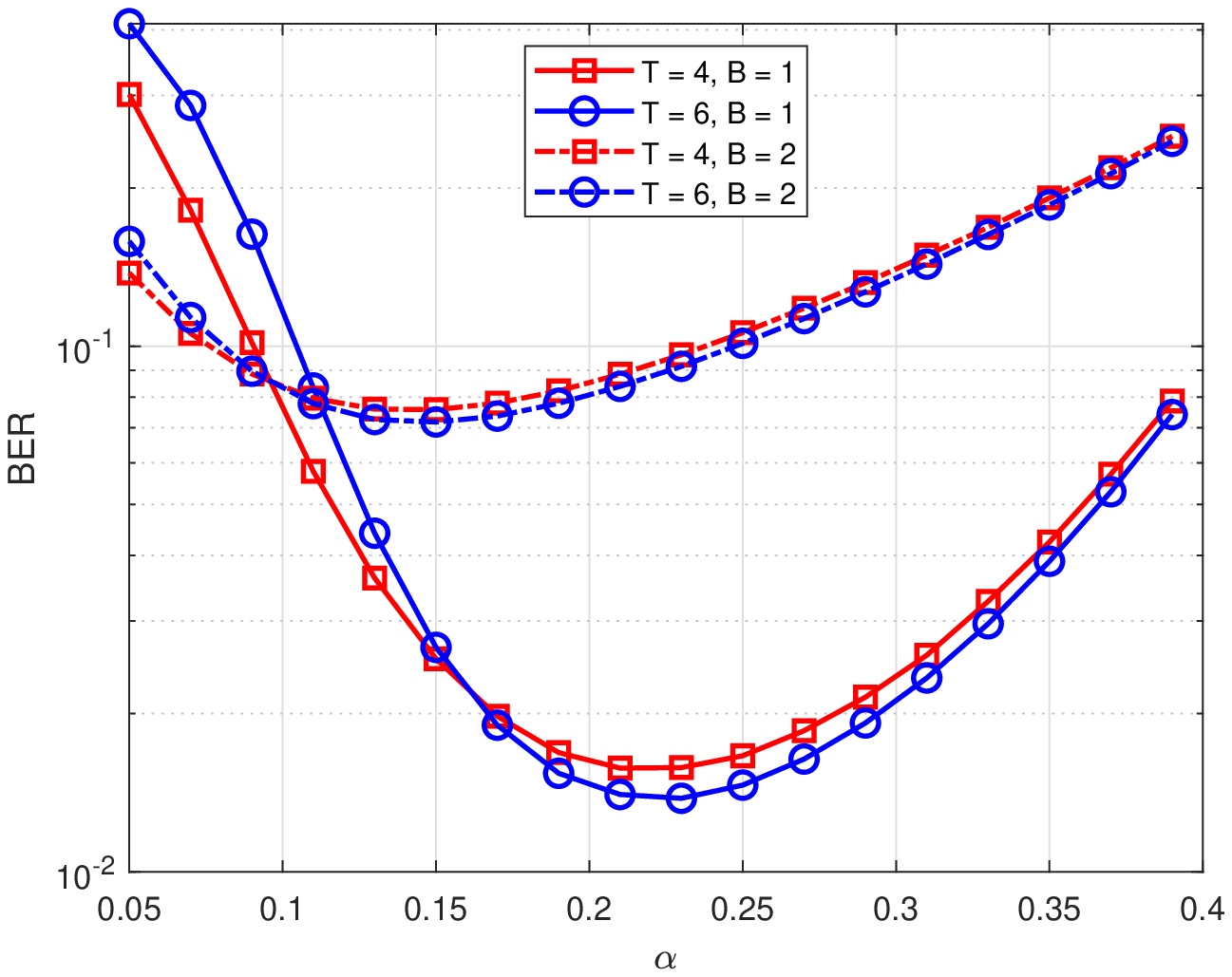}
        \caption{SNR = 10 dB}
        \label{fig:alpha_BER_SNR10}
        \end{subfigure}
        \begin{subfigure}{0.49\textwidth}
        \centering
        \includegraphics[width=\columnwidth]{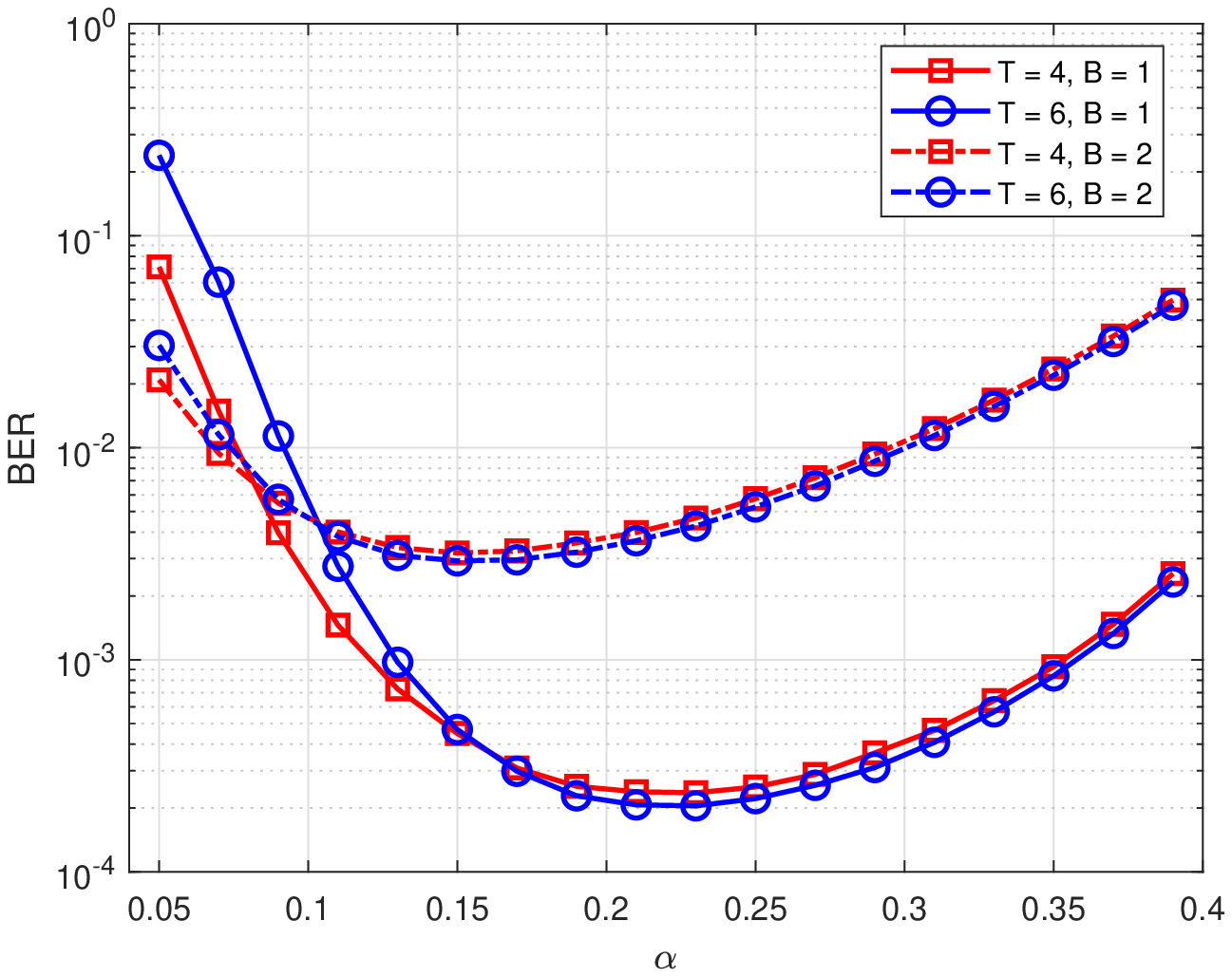}
        \caption{SNR = 20 dB}
        \label{fig:alpha_BER}
    \end{subfigure}\hfill
    \caption{BER as a function of $\alpha$ of the Grass-Lattice constellation for $T \in \{4,6\}$, $N = 2$, $B \in \{1,2\}$ and SNR $\in \{10,20\}$ dB.}
\end{figure}

Let us first evaluate the influence of $\alpha$, which determines the length of the lattice used for each real component in \eqref{eq:alpha}, on the SER and the BER. Figs. \ref{fig:alpha_SNR10}, \ref{fig:alpha}, \ref{fig:alpha_BER_SNR10} and \ref{fig:alpha_BER} show the SER/BER vs. $\alpha$ curves for SNR $\in \{10, 20\}$ dB, $T \in \{4,6\}$, $B \in \{1,2\}$ and $N = 2$. Remember that the spectral efficiency of the Grass-Lattice constellation is $R = \frac{2(T-1)B}{B}$ b/s/Hz. As we can see, $\alpha$ may have a significant impact on the SER and BER performance of the Grass-Lattice constellation. Further, the SER and BER vary significantly with the number of bits, $B$, used to encode each real component. It is also worth noticing that the SER/BER vs. $\alpha$ curves are smooth functions with a unique minimum so the optimal value $\alpha^*$ can be easily determined by searching over a predetermined grid. Clearly, the number of bits $B$ influences the optimal $\alpha^*$ more than the coherence time $T$. Another aspect that we observe is that the BER and SER vary similarly with $\alpha$, and hence the optimal value $\alpha^*$ can be obtained from either the SER or BER curve. We can also see that the optimal value $\alpha^*$ does not change significantly with the SNR. Therefore, for the rest of experiments in this section, we will choose the value of $\alpha$ that provides the lowest SER at SNR = $20$ dB. This value is easily precomputed offline and then used throughout the entire simulation.

% \begin{figure}[h!]
%     \begin{center}
%         \includegraphics[width=.7\linewidth]{images/alpha_B12_T46_N2.eps}
%     \end{center}
%     \caption{SER as a function of $\alpha$ of the Grass-Lattice constellation for $T \in \{4,6\}$, $N = 2$, $B \in \{1,2\}$ and SNR = $20$ dB.}
%     \label{fig:alpha}
% \end{figure}

% \begin{figure}[h!]
%     \begin{center}
%         \includegraphics[width=.77\linewidth]{images/alpha_BER_T46_N2_B12_SNR20.eps}
%     \end{center}
%     \caption{BER vs. $\alpha$ for the Grass-Lattice constellation for $T \in \{4,6\}$, $N = 2$, $B \in \{1,2\}$ and SNR = $20$ dB.}
%     \label{fig:alpha_BER}
% \end{figure}

% \begin{figure}[h!]
%     \begin{center}
%         \includegraphics[width=.77\linewidth]{images/alpha_SER_T46_N2_B12_SNR10.eps}
%     \end{center}
%     \caption{SER vs. $\alpha$ for the Grass-Lattice constellation for $T \in \{4,6\}$, $N = 2$, $B \in \{1,2\}$ and SNR = $10$ dB.}
%     \label{fig:alpha_SNR10}
% \end{figure}

% \begin{figure}[h!]
%     \begin{center}
%         \includegraphics[width=.77\linewidth]{images/alpha_BER_T46_N2_B12_SNR10.eps}
%     \end{center}
%     \caption{BER vs. $\alpha$ for the Grass-Lattice constellation for $T \in \{4,6\}$, $N = 2$, $B \in \{1,2\}$ and SNR = $10$ dB.}
%     \label{fig:alpha_BER_SNR10}
% \end{figure}

\subsection{Minimum chordal distance vs. $\alpha$}

For constellations of relatively small cardinality (up to 1024 codewords), instead of resorting to a SER/BER simulation to calculate the optimum value of parameter $\alpha$, we can generate the whole Grass-Lattice constellation and use the minimum chordal distance between codewords as the criterion to optimize $\alpha$. In this way, $\alpha^{*}$ can be computed much faster than with the SER/BER simulation proposed in the previous section.

Fig. \ref{fig:min_chordal} shows the minimum chordal distance between codewords for different values of $\alpha$ ranging from 0.02 to 0.4, $T = 2$ and $B \in \{2,3,4,5\}$. We can observe that the functions are smooth and have a clear maximum, which gives the value of $\alpha^{*}$. We can also see that, in this case, the optimum value of $\alpha$ with respect to the minimum chordal distance does not change significantly when we increase the number of bits $B$ used to encode each real component. As the SER simulation for a fixed SNR gives a better approximation of which is the optimum value of $\alpha$ in practice, for the rest of experiments in this section we will choose the value of $\alpha$ that provides the lowest SER at SNR = $20$ dB, as it was stated before.

\begin{figure}[t!]
    \begin{center}
        \includegraphics[width=.7\linewidth]{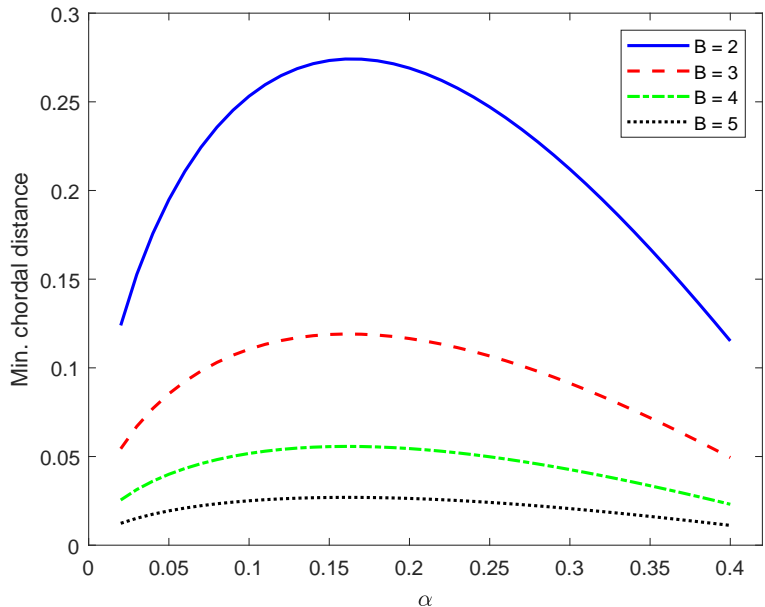}
    \end{center}
    \caption{Minimum chordal distance as a function of $\alpha$ of the Grass-Lattice constellation for $T = 2$ and $B \in \{2,3,4,5\}$.}
    \label{fig:min_chordal}
\end{figure}

\subsection{SER/BER vs. $E_b/N_0$}

Fig. \ref{fig:SER_T2} shows the SER as a function of $E_b / N_0$ for the proposed Grass-Lattice codebook for $T = 2$ symbol periods and $N = 1$ antenna. For comparison we include in the plot the structured  Cube-Split \cite{cube_split} and Exp-Map \cite{cipriano} constellations, as well as the unstructured Grassmannian constellations proposed in \cite{CuevasTCOM} that minimize the asymptotic PEP union bound (and hence labeled as UB-Opt). 

In addition, we include as a baseline the performance of a coherent pilot-based scheme. The transmitted signal for the pilot-based scheme is $\x_{coh} = [1, x_d]^T/\sqrt{2}$, where the first symbol is the constant pilot, which is known at the receiver, and the second symbol $x_d$ is taken from a QAM constellation with cardinality $2^{2(T-1)B}$, so that the coherent scheme has the same spectral efficiency as Grass-Lattice. That is, when $B=2$ we use a 16-QAM constellation, and when $B=3$ we use a 64-QAM constellation. The QAM constellations are normalized such that $E[|x_d|^2] =1$. Therefore, $E[\x_{coh}^H\x_{coh}]=1$ and hence the average transmit power of the pilot-based scheme is the same as that of the noncoherent schemes. Notice also that the power devoted to the data transmission is the same as the power devoted to training. This is the optimal power allocation for $T=2$ and $M=1$ as shown in \cite{Hassibi_TIT03}\footnote{In fact, it is shown in \cite{Hassibi_TIT03} that from an information-theoretic point of view using a number of pilots equal to the number of transmit antennas $M$ is always optimal, provided that we optimize the power allocation between pilots and data. Equal power allocation is optimal for $T=2$ and $M=1$. Nevertheless, we should bear in mind that these results are obtained by maximizing a lower bound on the capacity. Conclusions might be different if we optimize instead the SER or BER performance.}.

\begin{figure}[t!]
    \begin{center}
        \includegraphics[width=.7\linewidth]{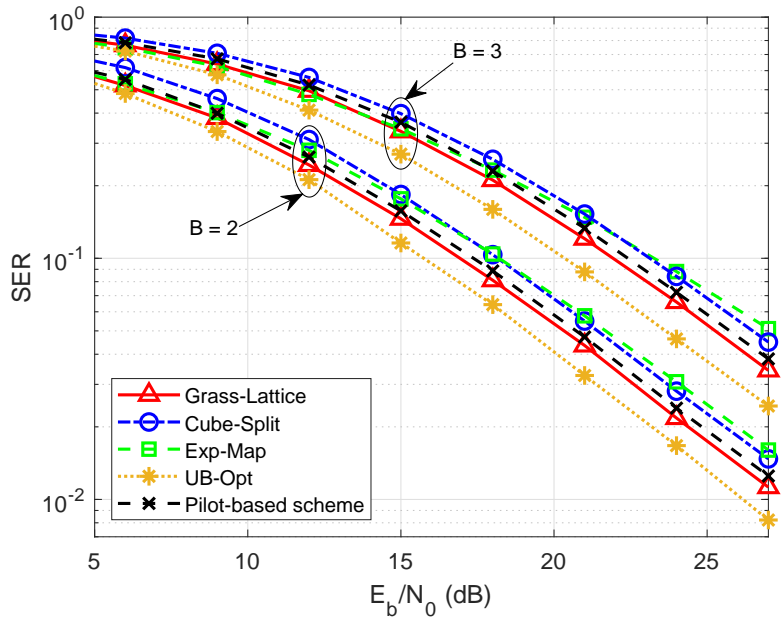}
    \end{center}
    \caption{Grass-Lattice SER curves in comparison with a pilot-based scheme, UB-Opt, Cube-Split and Exp-Map constellations for $T = 2$, $N = 1$ and $B \in \{2,3\}$.}
    \label{fig:SER_T2}
\end{figure}

For Grass-Lattice and Cube-Split we use $B \in \{2,3\}$ bits per real component, while for UB-Opt and Exp-Map we choose constellations with the same spectral efficiency as the ones provided by Grass-Lattice. In Fig. \ref{fig:SER_T2} we can observe that Grass-Lattice outperforms the other structured constellations and, as it was expected, it performs slightly worse than the unstructured UB-Opt constellation in terms of SER. Notice that UB-Opt uses the optimal ML detector in \eqref{eq:ML}, whereas Grass-lattice uses a suboptimal detector with much lower complexity.

\begin{figure}[htp!]
     \begin{subfigure}{\textwidth}
         \begin{center}
             \includegraphics[width=.7\columnwidth]{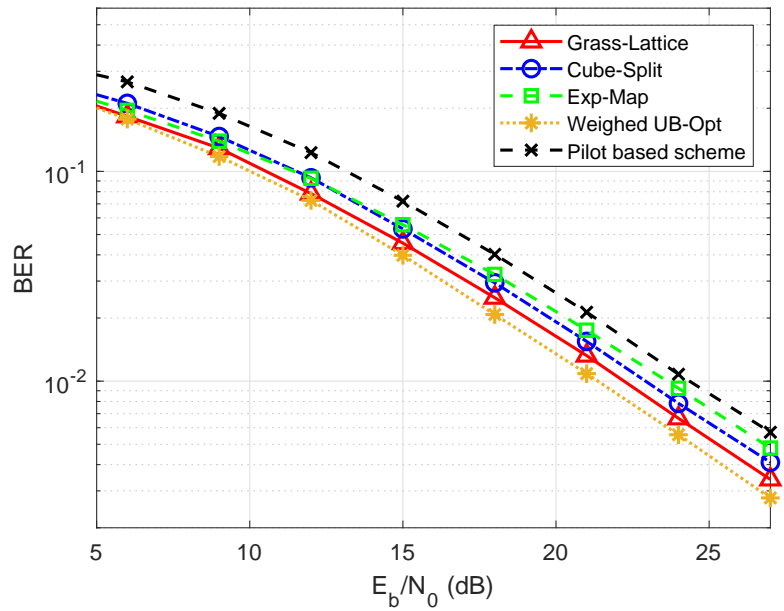}
         \end{center}
         \caption{$B = 2$}
         \label{fig:BER_T2_B2}
         \end{subfigure}
         \begin{subfigure}{\textwidth}
         \begin{center}
             \includegraphics[width=.7\columnwidth]{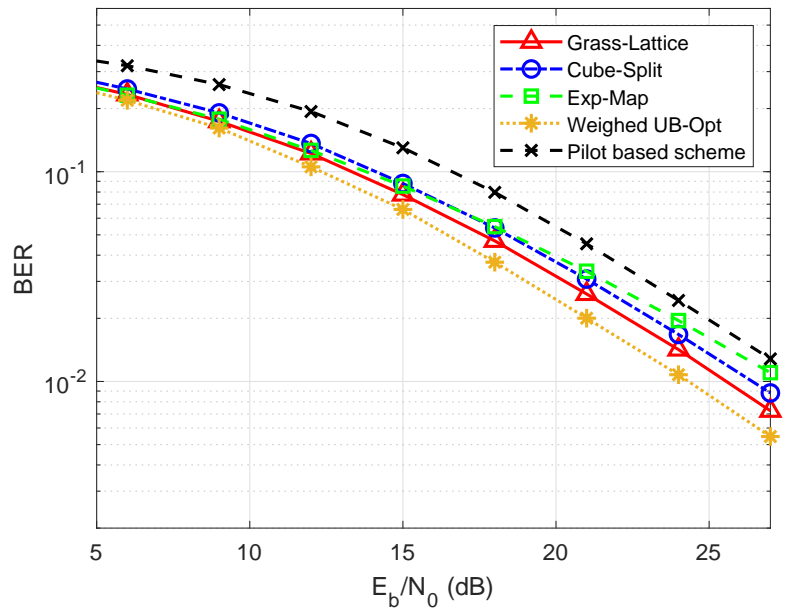}
         \end{center}
         \caption{$B = 3$}
         \label{fig:BER_T2_B3}
     \end{subfigure}
     \caption{Grass-Lattice BER curves in comparison with a pilot-based scheme, UB-Opt, Cube-Split and Exp-Map constellations for $T = 2$, $N = 1$ and $B \in \{2,3\}$.}
\end{figure}

Figs. \ref{fig:BER_T2_B2}, \ref{fig:BER_T2_B3} and \ref{fig:BER_T4} show the BER versus $E_b / N_0$ performance of Grass-Lattice constellations compared to Cube-Split, Exp-Map, weighed UB-Opt (joint constellation and bit-to-symbol mapping design) and a pilot-based scheme for $T = 2$, $N = 1$ and $B \in \{2,3\}$. For Grass-Lattice, we use a Gray encoding scheme that maps groups of $B$ bits to I/Q symbols defined in \eqref{eq:alpha}. A Gray-like encoder is also used for Cube-Split, Exp-Map and the pilot-based scheme. As we can see, Grass-Lattice constellations offer a superior performance in terms of BER than the other structured designs and the pilot-based scheme, which becomes more evident when the coherence time $T$ is smaller. The joint design of the unstructured constellation using the UB criterion and the bit labeling scheme provides for these examples the best performance.

% \begin{figure}[H]
%     \begin{subfigure}{\textwidth}
%         \begin{center}
%             \includegraphics[width=.7\columnwidth]{images/BER_comparison_T2_B23_N1_no_title.eps}
%         \end{center}
%         \caption{$T = 2$, $N = 1$}
%         \label{fig:BER_T2}
%         \end{subfigure}
%         \begin{subfigure}{\textwidth}
%         \begin{center}
%             \includegraphics[width=.7\columnwidth]{images/BER_comparison_T4_B123_N2_zoom_no_title.eps}
%         \end{center}
%         \caption{$T = 4$, $N = 2$}
%         \label{fig:BER_T4}
%     \end{subfigure}
%     \caption{Grass-Lattice BER curves in comparison with Cube-Split and Exp-Map constellations for $T \in \{2,4\}$, $N \in \{1,2\}$ and $B \in \{1,2,3\}$.}
%     \label{fig:BER_T24}
% \end{figure}

% \begin{figure}[H]
%     \begin{center}
%         \includegraphics[width=.6\linewidth]{images/BER_comparison_T2_B2_N1_vs_UBOpt_vs_coherent_no_title.eps}
%     \end{center}
%     \caption{Grass-Lattice BER curves in comparison with a pilot-based scheme, UB-Opt, Cube-Split and Exp-Map constellations for $T = 2$, $N = 1$ and $B = 2$.}
%     \label{fig:BER_T2_B2}
% \end{figure}

% \begin{figure}[H]
%     \begin{center}
%         \includegraphics[width=.6\linewidth]{images/BER_comparison_T2_B3_N1_vs_UBOpt_vs_coherent_no_title.eps}
%     \end{center}
%     \caption{Grass-Lattice BER curves in comparison with a pilot-based scheme, UB-Opt, Cube-Split and Exp-Map constellations for $T = 2$, $N = 1$ and $B = 3$.}
%     \label{fig:BER_T2_B3}
% \end{figure}

\begin{figure}[t!]
    \begin{center}
        \includegraphics[width=.7\linewidth]{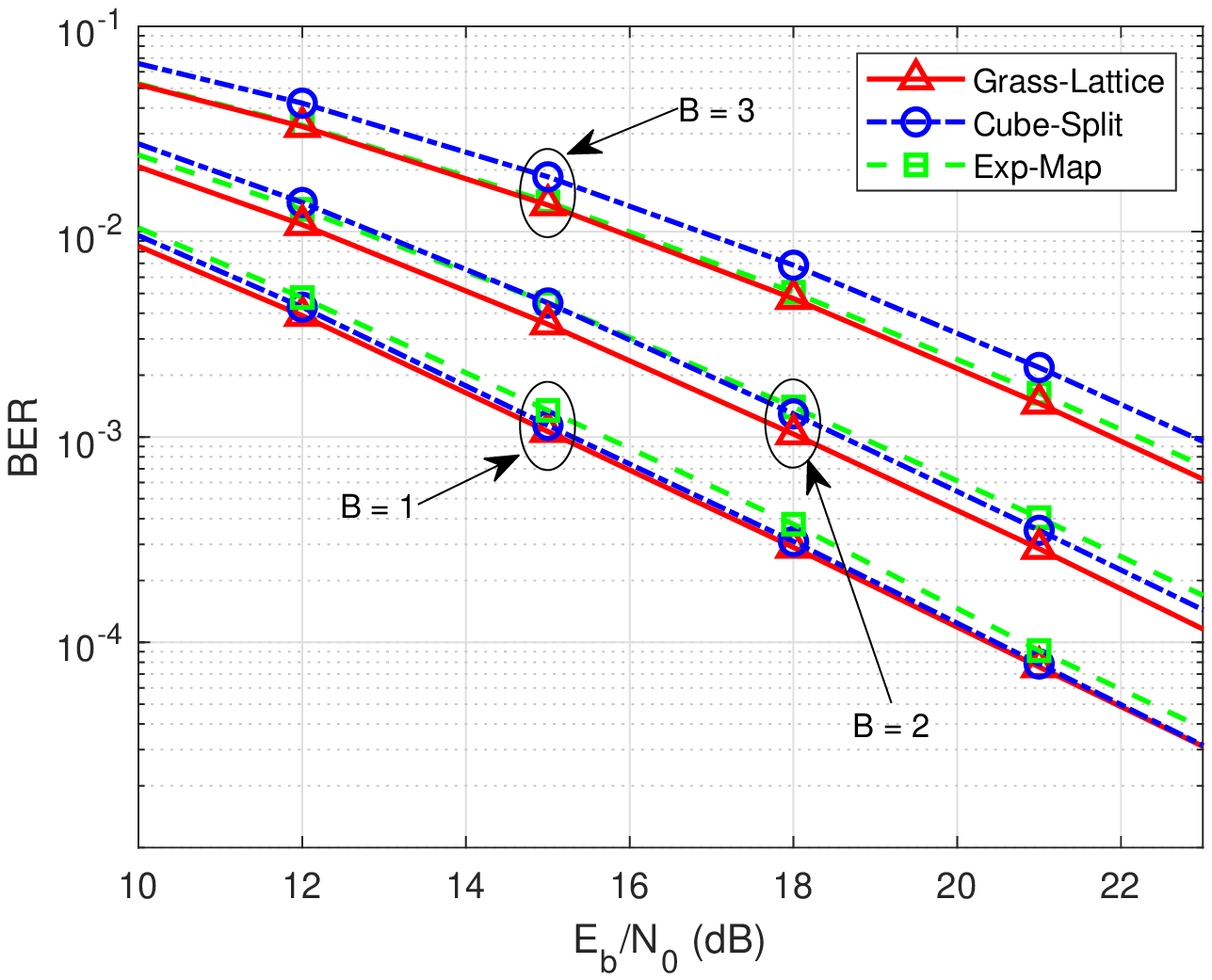}
    \end{center}
    \caption{Grass-Lattice BER curves in comparison with Cube-Split and Exp-Map constellations for $T = 4$, $N = 2$ and $B \in \{2,3\}$.}
    \label{fig:BER_T4}
\end{figure}

In Fig. \ref{fig:BER_T4} we consider a scenario with $T = 4$, $N = 2$ and $B \in \{1,2,3\}$. We restrict the comparison for this scenario to the Grass-Lattice, Cube-Split, and Exp-Map. Although Grass-Lattice is still the best performing method, the differences with Cube-Split are reduced, especially for a small number of bits.

\subsection{Spectral efficiency vs. $E_b / N_0$}

Finally, Fig. \ref{fig:spectral_eff} shows the spectral efficiency or rate in b/s/Hz against $E_b / N_0$ at BER=$10^{-4}$ for different values of $T$ and $N = 2$ for the Grass-Lattice and Cube-Split constellations. For given values of $T$ and $B$, the spectral efficiency of the Grass-Lattice code is  $\eta = \left(2B \left(T-1\right)\right) / \ T$ and the spectral efficiency of Cube-split is given by $\eta= \left(\log_2 T + 2B \left(T-1\right)\right) / \ T$. We notice from these two expressions that Cube-Split does not allow for a bit-to-symbol mapping when $T$ is not a power of 2, so Grass-Lattice achieves a wider range of spectral efficiencies. For example, we can see in this figure that Grass-Lattice allows you to design constellations for $T \in \{3, 6, 14\}$. For values of $T \in \{2, 4, 8\}$, for which Grass-Lattice and Cube-Split constellations can be both designed, we see that Grass-Lattice is more power efficient than Cube-Split when $T$ or $B$ grows. This could be at least partially explained by the fact that Cube-Split ignores the statistical dependencies between the different components of the codeword $\x$ for $T > 2$.

\begin{figure}[t!]
    \begin{center}
        \includegraphics[width=.7\linewidth]{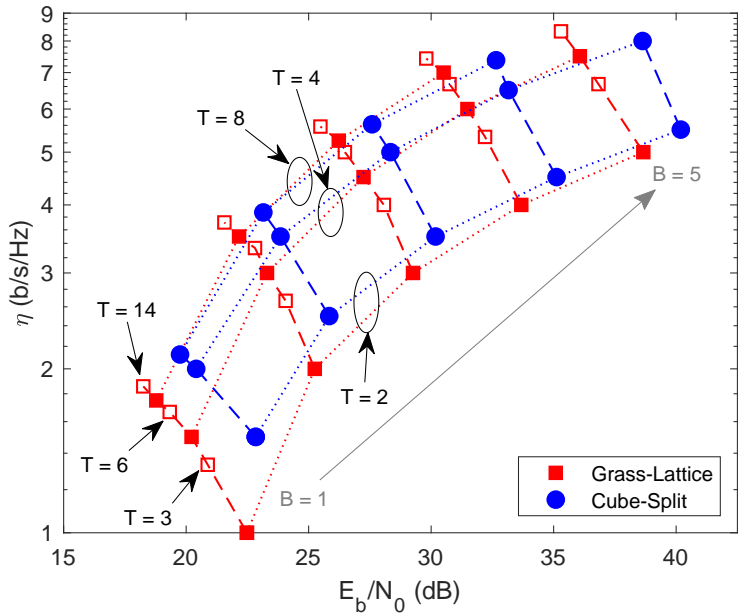}
    \end{center}
    \caption{Spectral efficiency of Grass-Lattice and Cube-Split as a function of $E_b / N_0$ at $10^{-4}$ BER for $T \in \{ 2,3,4,6,8,14 \}$ and $N = 2$.}
    \label{fig:spectral_eff}
\end{figure}

Fig. \ref{fig:spectral_eff_coherent} shows a comparison between Grass-Lattice and a coherent pilot-based scheme. For the coherent scheme for each value of $T$ we get three points that correspond to transmissions with 16-QAM, 32-QAM and 64-QAM signals. In all cases, the optimal number of pilots to send is 1, and the power allocation between the pilot and the data has been optimized according to \cite{Hassibi_TIT03}. We have used the MMSE channel estimator and the MMSE decoder. The figure clearly shows the spectral efficiency improvement of Grass-Lattice over the pilot-based scheme.

\begin{figure}[t!]
    \begin{center}
        \includegraphics[width=.7\linewidth]{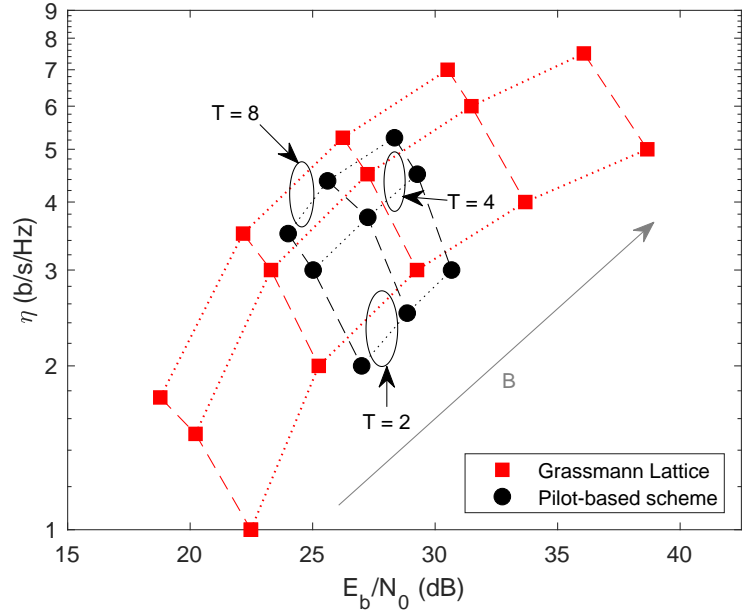}
    \end{center}
    \caption{Spectral efficiency of Grass-Lattice and a pilot-based scheme as a function of $E_b / N_0$ at $10^{-4}$ BER for $T \in \{ 2,4,8 \}$ and $N = 2$.}
    \label{fig:spectral_eff_coherent}
\end{figure}

%% file: section6_conclusions.tex
\section{Conclusion} \label{sec:conclusions}
We have proposed a new Grassmannian constellation for noncoherent communications in SIMO channels, named Grass-Lattice, based on a measure preserving mapping from the unit hypercube to the Grassmannian of lines. Thanks to its structure, the encoding and decoding steps can be performed on the fly with no need to store the whole constellation. Further, it allows for low-complexity and efficient decoding as well as for a simple Gray-like bit labeling. Simulation results show that Grass-Lattice has symbol and bit error rate performance close to that of a numerically optimized unstructured constellation. Besides, the designed constellations outperform other structured constellations in the literature and a coherent pilot-based scheme in terms of SER and BER under Rayleigh block fading channels, in addition to being more power efficient. As mappings $\vartheta_1$ and $\vartheta_3$ have already been derived in this paper for any number of transmit antennas, further research will be done to study the extension of mapping $\vartheta_2$ and, consequently, the whole Grass-Lattice mapping, to the MIMO case.

%% file: appendix1.tex
\subsection{Proof of Lemma \ref{lem:pointsunitball}}\label{app2}
Let us define $d= T-1$. The function $f_d$ is the unique solution of
$$
f(t)^{2d-1}(f(t)+tf'(t))=\frac{e^{-t^2}}{\Gamma(d+1)},\quad \lim_{t\to\infty}tf(t)=1,
$$
which satisfies $tf(t)\in[0,1)$ and can be written in terms of an incomplete Gamma function. It is easy to see that $\vartheta_2:\mathbb C^d\to\mathbb B_{\mathbb{C}^{d}}(0,1)$ is a diffeomorphism. Let us compute the Jacobian of $\vartheta_2$: if $\dot \z$ is (real) orthogonal to $\z$ then
$$
D\vartheta_2(\z)\dot \z=\dot \z f_d(\|\z\|),
$$
while for $\dot \z=\z/\|\z\|$ we have
\begin{equation*}
D\vartheta_2(\z)\frac{\z}{\|\z\|} = \frac{\z}{\|\z\|} f_d(\|\z\|)+\z f_d'(\|\z\|)
=\z\left(\frac{f_d(\|\z\|)}{\|\z\|}+f_d'(\|\z\|)\right).
\end{equation*}
Choosing any orthonormal basis of $\mathbb C^d\equiv \mathbb R^{2d}$ whose last vector is $\z/\|\z\|$ we thus have that the orthogonality of this basis is preserved by $D\vartheta_2$. The Jacobian of $\vartheta_2$ at $\z$ is then just the product of the lengths of the resulting vectors:
$$
Jac \, \vartheta_2(\z)=f_d(\|\z\|)^{2d-1}(f(\|\z\|)+\|\z\|f'(\|\z\|))=\frac{e^{-\|\z\|^2}}{\Gamma(d+1)}.
$$

\noindent Given any integrable mapping $g:\mathbb B_{\mathbb{C}^d}(0,1)\to\mathbb R$, the expected value of $g(\w)$ when $\w$ follows the distribution of the lemma is:
\begin{equation*}
I = \frac{1}{\pi^d}\int_{\z\in\mathbb C^d}g(\vartheta_2(\z))e^{-\|\z\|^2}\,d\z
=\frac{\Gamma(d+1)}{\pi^d}\int_{\z\in\mathbb C^d}g(\vartheta_2(\z))Jac \, \vartheta_2(\z)\,d\z,
\end{equation*}
which by the Change of Variables Theorem equals
$$
\frac{\Gamma(d+1)}{\pi^d}\int_{\w\in\mathbb B_{\mathbb{C}^d}(0,1)}g(\w)\,d\w.
$$
This is the expected value of $g$ in $\mathbb B_{\mathbb{C}^d}(0,1)$, since the volume of $\mathbb B_{\mathbb{C}^d}(0,1)$ is precisely $\pi^d/\Gamma(d+1)$.

%% file: appendix4.tex
\subsection{Proof of Lemma \ref{lem:otrojacobian}}\label{app4}

That the formula for $\Theta^{-1}$ is the claimed one is easy to see: just write down the singular value decomposition of $\W= \U \binom{{\bf D}} {{\bf 0}}\V^\He$ and compose the two functions in any order to see that you get the identity map in each space. Now let us compute the Jacobian. First, note that for any given unitary $T\times T$ matrix $\U$ the isometry $\A\to \U\A$ in the domain commutes with the isometry $\B\to \U\B$ in the range, and the same happens with the isometry $\A\to\A \V$ if $\V$ is unitary of size $M$. It suffices to prove our result in the case that $\A=\binom{\D}{{\bf 0}}$ with $\D=\diag(\sigma_1,\ldots,\sigma_M)$. Let us compute the corresponding directional derivatives:
\begin{itemize}
    \item For $\dot\A=\binom{{\bf 0}}{\dot \C}$ we have
    \begin{multline*}
    D\Theta(\A)(\dot \A)=\frac{d}{dt}\mid_{t=0}\left(\binom{\D}{t\dot\C}\left(\I_M+ \begin{pmatrix} \D &t\dot\C^\He \end{pmatrix}\binom{\D}{t\dot\C}\right)^{-1/2}\right)= \\\binom{{\bf 0}}{\dot\C}(\I_M+\D^2)^{-1/2}.
    \end{multline*}
    The natural basis for $\dot\C$ then preserves orthogonality and this yields a factor for the Jacobian of $\Theta$ of:
    $$
    \prod_{m=1}^M\frac{1}{(1+\sigma_m^2)^{T-2M}}=\det(\I_M+\A^\He\A)^{-T+2M}.
    $$
    \item If $\dot \A=\binom{\deltaB_{kk}}{{\bf 0}}$, where $\deltaB_{kk}$ denotes an $M\times M$ matrix whose only nonzero term $\delta_{kk}$ is located at row $k$ and column $k$, then a direct computation shows that
    $$
    D\Theta(\A)(\dot \A)=\frac1{(1+\sigma_k^2)^{3/2}},
    $$
    and similarly if $\dot \A=\binom{j\delta_{kk}}{{\bf 0}}$ then
    $$
    D\Theta(\A)(\dot \A)=j\frac1{(1+\sigma_k^2)^{1/2}},
    $$
    which again preserves orthogonality and adds the following factor to the Jacobian of $\Theta$
    $$
    \prod_{m=1}^M\frac1{(1+\sigma_m^2)^2}=\det(\I_M+\A^\He\A)^{-2}
    $$
    \item If $\dot\A=\binom{\deltaB_{12}}{{\bf 0}}$, denoting $R=\sqrt{(1+\sigma_1^2)(1+\sigma_2^2)}$ then we have
    \begin{align*}
    D\Theta(\A)(\dot \A)=&\frac{d}{dt}\mid_{t=0} \left(\binom{\D+t\delta_{12}}{0}\left(\I_M+ \left(\D+t\delta_{21}\right)\left(\D+t\delta_{12}\right)\right)^{-1/2}\right)\\
    =&
    \frac{1}{R\sqrt{2+\sigma_1^2+\sigma_2^2+2R}}\begin{pmatrix}
    0&1+R&&\\
    -\sigma_1\,\sigma_2&0&&\\
    &&0&&\\
    &&&\ddots&\\
    &&&&0
    \end{pmatrix},
    \end{align*}
    while if $\dot\A=\binom{\deltaB_{21}}{0}$ then we have
       \begin{align*}
    D\Theta(\A)(\dot \A)=&\frac{d}{dt}\mid_{t=0} \left(\binom{\D+t\delta_{21}}{0}\left(\I_M+ \left(\D+t\delta_{12}\right)\left(\D+t\delta_{21}\right)\right)^{-1/2}\right)\\
    =&
    \frac{1}{R\sqrt{2+\sigma_1^2+\sigma_2^2+2R}}\begin{pmatrix}
    0&-\sigma_1\,\sigma_2&&\\
    1+R&0&&\\
    &&0&&\\
    &&&\ddots&\\
    &&&&0
    \end{pmatrix},
    \end{align*}
    Hence, the volume of the parallelepiped spanned by these two vectors is
    $$
    \frac{(1+R)^2-\sigma_1^2\sigma_2^2}{R^2(2+\sigma_1^2+\sigma_2^2+2R}=\frac{1}{R^2}.
    $$
    This yields a factor $\frac{1}{(1+\sigma_1^2)(1+\sigma_2^2)}$ for the Jacobian. The same computation for $\deltaB_{ij}$ gives all together:
    $$
    \prod_{m=1}^M\frac{1}{(1+\sigma_m^2)^{M-1}}=\det(\I_M+\A^\He\A)^{-(M-1)}
    $$
    \item If $\dot\A=\binom{\mathrm{j} \deltaB_{12}}{{\bf 0}}$ and later $\dot\A=\binom{\mathrm{j} \deltaB_{21}}{{\bf 0}}$ we get the same computation, which yields another factor of
        $$
    \prod_{m=1}^M\frac{1}{(1+\sigma_m^2)^{M-1}}=\det(\I_M+\A^\He\A)^{-(M-1)}.
    $$
\end{itemize}
Multiplying all the factors, we have that the Jacobian of $\Theta$ is  $\det(\I_M+\A^\He\A)^{-T}$. This finishes the proof.

%% file: appendix3.tex
\subsection{Proof of Theorem \ref{th:genproj}}\label{app_th1}
Let $G:\mathbb{G}\left(1,\mathbb{C}^T\right)\to\mathbb C$ be integrable. From Lemma \ref{cor:integrales2},
\begin{equation*}
\frac{1}{Vol(\mathbb{G}\left(1,\mathbb{C}^T\right))} \int_{[\x]\in\mathbb G(1,\mathbb C^{T})}G([\x])\,d[\x]=\frac{1}{Vol(\mathbb B_{\mathbb C^{T-1}}(0,1))}\int_{\underset{\|\w\|<1}{\w\in \mathbb C^{T-1}} }G\left( \begin{bmatrix} \sqrt{1-\|\w\|^2} \\ \w \end{bmatrix} \right)\,d\w,
\end{equation*}
where $Vol(\mathbb{S})$ denotes the volume of the set $\mathbb{S}$. From Lemma \ref{lem:pointsunitball}, this equals
$$
\frac{1}{\pi^{T-1}} \int_{\z\in\mathbb C^{T-1}} G\left( \begin{bmatrix} \sqrt{1-\|\z f_{T-1}(\|\z\|)\|^2} \\ \z f_{T-1}(\|\z\|) \end{bmatrix} \right) e^{-\|\z\|^2}\,d\z,
$$
which in turn from Lemma \ref{lem:twofold} equals
$$
\int_{(\a,\b)\in\mathcal{I}}G\left( \begin{bmatrix} \sqrt{1-\|\z f_{T-1}(\|\z\|)\|^2} \\ \z f_{T-1}(\|\z\|) \end{bmatrix} \right) \,d(\a,\b),
$$
where $$\z=(z_1,\ldots,z_{T-1})^{\textnormal{T}}, \quad z_k=F^{-1}(a_k)+ j F^{-1}(b_k).$$
All in one, we have proved that the point
\begin{equation*}
\begin{bmatrix} \sqrt{1-\|\w\|^2} \\ \w \end{bmatrix}
\end{equation*}
with $\w = \z f_{T-1}(\|\z\|)$, is uniformly distributed in $\mathbb{G}\left(1,\mathbb{C}^T\right)$.